\documentclass{article}
\usepackage{amssymb,amsthm,amsmath}
\usepackage{tkz-graph}
\usepackage{graphics}
\usepackage[utf8]{inputenc}
\newtheorem{theorem}{Theorem}
\newtheorem{lemma}[theorem]{Lemma}
\theoremstyle{definition}

\newtheorem*{remark}{Remark}
\newtheorem{example}[theorem]{Example}

\newenvironment{keywords}{{\bf Keywords:}}{}
\title{A Note on the Inheritance of the Isometry-Dual Property under Puncturing AG Codes}
\author{Maria Bras-Amorós}
\begin{document}
\maketitle

\begin{abstract}
Consider a sequence of AG codes evaluating at a set of
evaluation points $P_1,\dots,P_n$ the functions having only poles at a
defining point $Q$, with the sequence of codes satisfying the
isometry-dual condition (i.e. containing at the same time primal and their dual codes). We prove a necessary condition under which,
after taking out a number of evaluation points (i.e. puncturing), the
resulting AG codes can still satisfy the isometry-dual property. 
The condition has to do with the so-called maximum sparse ideals of
the Weierstrass semigroup of $Q$. 
\end{abstract}


\section{Introduction}
Two codes $C,D\subseteq{\mathbb F}_q^n$ are said to be $x$-isometric, for $x\in{\mathbb F}_q^n$ if and only if the map $\chi_x:{\mathbb F}_q^n\rightarrow{\mathbb F}_q^n$ given by the component-wise product $\chi_x(v)=x * v$ satisfies 
$\chi_x(C)=D$. Then, a sequence of codes $(C^i)_{i={0,\dots,n}}$ is said to satisfy the {\it isometry-dual condition} if there exists $x\in({\mathbb F}_q^*)^n$ such that $C^i$ is $x$-isometric to $(C^{n-i})^\perp$ for all $i=0,1,\dots,n$.
Sequences of one-point AG codes satisfying the isometry-dual condition
were characterized in \cite{GMRT} in terms of the Weierstrass
semigroup at the defining point. The result can be stated in terms of
maximum sparse ideals of numerical semigroups, that is, ideals whose gaps
are maximal. 

In this contribution we analyze the effect of puncturing
a sequence of AG codes satisfying the isometry-dual property in terms
of the inheritance of this property.

In Section 2 we introduce and characterize maximum sparse ideals.
In Section 3 we analyze the inclusion among maximum sparse ideals. In
Section 4 we prove that the set of maximum sparse ideals is in bijection with
another ideal. In Section 5 we prove that the isometry-dual property can be inherited after puncturing a sequence of
codes only if the number of punctured coordinates is a non-gap of the
Weierstrass semigroup at the defining point. In particular, the
property is not inherited in general if one only takes out one coordinate.

\section{Maximum sparse ideals}

A {\em numerical semigroup} $S$ is a subset of ${\mathbb N}_0$ that contains $0$, is closed under addition and has a finite complement in ${\mathbb N}_0$.
An {\em ideal} $I$ of a numerical semigroup $S$ is a subset of $S$
such that $I+S\subseteq I$. We say that $I$ is a {\em proper} ideal of $S$ if $I\neq S$.

Denote the elements of $S$, in increasing ordre, by $\lambda_0=0,
\lambda_1,\dots$ and call {\em genus} the number $g=\#{\mathbb
  N}_0\setminus S$. The {\em conductor} $c$ of the semigroup is the smallest integer such that $c+{\mathbb N}_0\subseteq S$.

It is proved in \cite{BLV} that the largest integer not belonging to
an ideal, which is called the {\em Frobenius number} of the ideal is at most $2g-1+\# (S\setminus I)$.
The ideals whose Frobenius number attains this bound will be called {\it maximum sparse ideals}. 

It is also proved in \cite{BLV} that, letting $G(i)$ be the number of pairs of gaps adding up to $\lambda_i$, a proper ideal $I$ is maximum sparse if and only if $I=S\setminus D(i)$ for some $i$ with $G(i)=0$. In this case we denote $\lambda_i$ the leader of the maximum sparse ideal $I$ and it coincides with the maximum gap of the ideal.

\begin{lemma}
The leaders of proper maximum sparse ideals are always at least the conductor.
\end{lemma}

\begin{proof}
  The Frobenius number of the maximum sparse ideal is $2g-1+\#(S\setminus I)$, which, since the ideal is $S\setminus D(i)$, equals the maximum of $D(i)$. But the maximum of $D(i)$ is $\lambda_i$, that is, the leader of the ideal.
  Then $\lambda_i=2g-1+\#(S\setminus I)$. Since the ideal is proper, $\#(S\setminus I)\geq 1$, and so $2g-1+\#(S\setminus I)\geq 2g\geq c$. Hence, $\lambda_i\geq c$.
\end{proof}

\section{Inclusion among maximum sparse ideals}

\begin{lemma}
\label{l:diffnongaps}
If the proper ideals $I,I'$ are maximum sparse, with leaders $\lambda_i,\lambda_{i'}$, and $I'\supseteq I$, then $\lambda_i-\lambda_{i'}\in S$.
\end{lemma}

\begin{proof}
The inclusion $I\supseteq I$ implies $(S\setminus I')\subseteq (S\setminus I)$. We know that $(S\setminus I)=D(i)$, $(S\setminus I')=D(i')$, so, $D(i')\subseteq D(i)$. This is equivalent to $\lambda_{i'}\in D(i)$ which means that $\lambda_i-\lambda_i'\in S.$
  \end{proof}

\begin{lemma}
\label{l:diffcards}
If the proper ideals $I,I'$ are maximum sparse and $I'\supseteq I$, then $\#(S\setminus I)-\#(S\setminus I')\in S$.
\end{lemma}

\begin{proof}
Suppose that the leaders of $I,I'$ are $\lambda_i,\lambda_{i'}$.
  Since both $I$ and $I'$ are maximum sparse, $\#(S\setminus I)=(\lambda_i-(2g-1))$ and $\#(S\setminus I')=(\lambda_{i'}-(2g-1))$. Consequently,
$\#(S\setminus I)-\#(S\setminus I')=\lambda_i-\lambda_{i'}$ which, by Lemma~\ref{l:diffnongaps}, belongs to $S$.
\end{proof}

\begin{remark}
The converse is not true in general. Suppose that $\lambda_i$ is the leader of a maximum sparse ideal and that $\lambda_{i'}$, which is at least the conductor, satisfies $\lambda_i-\lambda_{i'}\in S$. This does not imply that $\lambda_{i'}$ is the leader of any maximum sparse ideal, unless $D(i')=0$.
\end{remark}

The previous results lead to the next theorem.

\begin{theorem}
  For two proper sparse ideals $I,I'$ with leaders $\lambda_i,\lambda_{i'}$, the following are equivalent:
  \begin{itemize}
  \item $I'\supseteq I$
  \item $\lambda_i-\lambda_{i'}\in S$.
  \item $S\setminus I'\subseteq S\setminus I$.
  \item $\#(S\setminus I)-\#(S\setminus I)\in S$.
  \end{itemize}
\end{theorem}

\section{The ideal of sparse ideal leaders}

\begin{lemma} 
The set $L$ of non-zero non-gaps $\lambda_i$ such that $G(i)=0$ is an ideal of $S$.
\end{lemma}
  
\begin{proof}
First of all, notice that the non-gaps smaller than the conductor do not satisfy $G(i)=0$. Indeed, if $\lambda_i<c$, there is a gap $a<\lambda_i$ with $\lambda_i-a<\lambda_1$ because otherwise $\lambda_i$ would be larger than the conductor. Now, $\lambda_i-a<\lambda_1$ is a positive gap which, together with $a$ adds up to $\lambda_i$. So, $G(i)\neq 0$. So, all the elements in $L$ are at least equal to the conductor.
  
  We need to prove that if $\lambda_i\in L$ then $\lambda_i+\lambda_j\in L$ for any $\lambda_j\in S$. We can assume that $\lambda_j\neq 0$ because otherwise it is obvious. Let $k$ be such that $\lambda_i+\lambda_j=\lambda_k$.

Suppose that $\lambda_k\not\in L$ and so that $G(k)\neq 0$. This means that there exist two gaps $a$, $b$ such that $a+b=\lambda_k$. Since $\lambda_k=\lambda_i+\lambda_j$ with $\lambda_i\geq c$, we have $\lambda_k-1,\lambda_k-2,\lambda_k-3,\dots,\lambda_k-\lambda_j\in S$ since they all are larger than $\lambda_i$. Hence both $a$ and $b$ are smaller than $\lambda_k-\lambda_j$ and so, since $a+b=\lambda_k$, they both are larger than $\lambda_j$. Then $a-\lambda_j$ is a gap of $S$ since, otherwise, $a=(a-\lambda_j)+\lambda_j\in S+S\subseteq S$.
In  particular, $(a-\lambda_j)+b$ is a sum of two gaps which adds up to $a+b-\lambda_j=\lambda_k-\lambda_j=\lambda_i$, a contradiction since $\lambda_i$ is assumed to belong to $L$ and so $D(i)=0$.
\end{proof}

\section{Puncturing sequences of isometry-dual one-point AG codes}

Let $P_1,\dots,P_n, Q$ be different rational points of a (projective, non-singular, geometrically irreducible) curve with genus $g$
and define $C_m=\{(f(P_1),\dots,f(P_n)):f\in L(mQ)\}$.
Note that it can be the case that $C_m=C_{m-1}$.
Let $W=\{m\in{\mathbb N}: L(mP)\neq L((m-1)P)\}$ be the Weierstrass
semigroup at $Q$ and let $W^*=\{0\}\cup\{m\in{\mathbb N}, m>0:C_m\neq
C_{m-1}\}=\{m_1=0,m_2,\dots,m_n\}$. 
Then, if $n>2g+2$, the set $W\setminus W^*$ is an ideal of $W$ (this is stated in different words in \cite[Corollary 3.3.]{GMRT}).
In particular, $C^0=\{0\}$ together with $C^1=C_{m_1},C^2=C_{m_2},\dots,C^n=C_{m_n}$ satisfy the isometry-dual condition if and only if $n+2g-1\in W^*$, that is, if and only if $W\setminus W^*$ is maximum sparse. This is proved in \cite[Proposition 4.3.]{GMRT}.

\begin{theorem}\label{t:inh}Suppose that the sequence $C^0,C_{m_1},C_{m_2},\dots,C_{m_n}$ as just defined satisfies the isometry-dual condition.
Let $\{P_{i_1},\dots,P_{i_{n'}}\}\subseteq\{P_1,\dots,P_n\}$,
with $2g+2<n'<n$. Define $C'_m=\{(f(P_{i_1}),\dots,f(P_{i_{n'}})):f\in
L(mQ)\}$
and $(W^*)'=\{0\}\cup\{m\in{\mathbb N}, m>0:C'_m\neq C'_{m-1}\}=\{m'_1=0,m'_2,\dots,m'_{n'}\}$.

If the sequence $\{0\},C'_{m'_1},C'_{m'_2},\dots,C_{m'_{n'}}$
satisfies the isometry-dual condition, then $n-n'\in W$.
\end{theorem}

\begin{proof}
  Since the sequence $C^0,C_{m_1},C_{m_2},\dots,C_{m_n}$ satisfies the
  isometry-dual condition, the ideal $W\setminus W^*$ is maximum
  sparse. 

If the sequence
$\{0\},C'_{m'_1},C'_{m'_2},\dots,C_{m'_{n'}}$ also satisfies the
isometry-dual condition, then $W\setminus (W^*)'$ is also sparse.
  
Since $C'_m\neq C'_{m-1}$ implies $C_m\neq C_{m-1}$, we have $(W^*)'\subseteq W^*$
and so $W\setminus (W^*)'\supseteq S\setminus W^*$.
Now, by Lemma~\ref{l:diffcards}, $\# W^*-\# (W^*)'=n-n'\in W$.
\end{proof}

The conclusion of Theorem~\ref{t:inh} is that, given a sequence of AG
codes evaluating at a set of evaluation points $P_1,\dots,P_n$ the
functions having only poles at a defining point $Q$, with the sequence
satisfying the isometry-dual condition, one needs to take out a number
of evaluation points at least equal to the multiplicity (smallest
non-zero non-gap) of the Weierstrass semigroup of $Q$ in order to obtain another punctured sequence satisfying the isometry-dual property.

\begin{example}
  The Hermitian curve over ${\mathbb F}_{q^2}$, with affine equation $x^{q+1}=y^q+y$, has $q^3$ affine points and one point $P_\infty$ at infinity.
  A basis of $\cup_{m\geq 0}L(mP_\infty)$ is given in increasing order by the list
  $\left((x^{d-i}y^i)_{i\in \{0,\dots,\min(q-1,d)\}}\right)_{d\geq 0}$.

  Consider the Hermitian curve over ${\mathbb F}_4$.
  For this, let $\alpha$ be the class of $x$ in ${\mathbb F}_4=({\mathbb Z}_2[x])/(x^2+x+1)$.

  The $8$ points of the Hermitian curve over ${\mathbb F}_4$ are 
  $P_1=(0, 0), P_2=(\alpha,\alpha), P_3=(\alpha + 1, \alpha), P_4=(1, \alpha), P_5=(\alpha, \alpha + 1), P_6=(\alpha + 1, \alpha + 1), P_7=(1, \alpha + 1), P_8=(0, 1)$. The curve has genus $g=1$.

  The Weierstrass semigroup at $P_\infty$ is $W=\{0,2,3,4,5,6,\dots\}$ and one basis for $\cup_{m\geq 0}L(mP_\infty)$ is 
$${\mathcal B}=\{1, x, y, x^2, xy, x^3, x^2y, x^4, x^3y, x^5, x^4y, x^6, x^5y, x^7, x^6y, x^8, x^7y,\dots\}.$$

For each subset of points $\mathcal P\subseteq\{P_1,\dots,P_8\}$ we analyzed whether $\#{\mathcal P}+2g-1\in(W^*)'$,
that is, if the evaluation vector of the $(\#{\mathcal P}+g)$th function of ${\mathcal B}$ at the points of ${\mathcal P}$ is 
not linearly dependent of the set of the evaluation vectors of the previous functions in ${\mathcal B}$ at the points of ${\mathcal P}$.

In Figure~\ref{fig} we depicted all sets of points satisfying $\#{\mathcal P}+2g-1\in(W^*)'$, denoting the set $P_{i_1},P_{i_2},\dots,P_{i_s}$ by just
$i_1i_2\dots i_s$. We draw an edge for every inclusion relation among sets of points. 

According to the assumption $n>2g+2$ we draw a line separating the sets of point such that $\#{\mathcal P}>4$. For the sets of points on the left of this line,  being in the graph is equivalent to satisfying the isometry-dual condition (by \cite[Proposition 4.3.]{GMRT}).

One can check that there are only edges between sets of points whose cardinality difference is at least two. This is coherent with Theorem~\ref{t:inh} and the fact that $W=\{0,2,3,4,5,6,\dots\}$.

\end{example}

\begin{figure}
  \caption{Hierarchy of sets of points of the Hermitian curve over ${\mathbb F_4}$.}
  \label{fig}
\resizebox{\textwidth}{!}{
\begin{tikzpicture}
{\Huge\bfseries
\SetGraphUnit{3.75}
             {    \Vertex[x=0,y=0]{12345678}
    \Vertices[x=10,y=0,dir=\SO]{line}{123568,124578,134678,234567}
    \Vertices[x=20,y=30,dir=\SO]{line}{12348,12357,12467,13456,15678,23678,24568,34578}
    \draw[dashed] (25, 35) -- (25, -25);
    \Vertices[x=30,y=0,dir=\SO]{line}{1258,1368,1478,2356,2457,3467}
    \Vertices[x=40,y=30,dir=\SO]{line}{126,137,145,234,278,358,468,567}
    \Vertices[x=50,y=0,dir=\SO]{line}    {18,47,36,25}
  }
  \Edge(18)(1258)
  \Edge(18)(1368)
  \Edge(18)(1478)
  \Edge(25)(1258)
  \Edge(25)(2356)
  \Edge(25)(2457)
  \Edge(36)(1368)
  \Edge(36)(2356)
  \Edge(36)(3467)
  \Edge(47)(1478)
  \Edge(47)(2457)
  \Edge(47)(3467)
  \Edge(18)(12348)
  \Edge(18)(15678)
  \Edge(25)(12357)
  \Edge(25)(24568)
  \Edge(36)(13456)
  \Edge(36)(23678)
  \Edge(47)(12467)
  \Edge(47)(34578)
  \Edge(126)(12467)
  \Edge(137)(12357)
  \Edge(145)(13456)
  \Edge(234)(12348)
  \Edge(278)(23678)
  \Edge(358)(34578)
  \Edge(468)(24568)
  \Edge(567)(15678)
  \Edge(126)(123568)
  \Edge(358)(123568)
  \Edge(145)(124578)
  \Edge(278)(124578)
  \Edge(137)(134678)
  \Edge(468)(134678)
  \Edge(234)(234567)
  \Edge(567)(234567)
  \Edge(1258)(123568)
  \Edge(1368)(123568)
  \Edge(2356)(123568)
  \Edge(1258)(124578)
  \Edge(1478)(124578)
  \Edge(2457)(124578)
  \Edge(1368)(134678)
  \Edge(1478)(134678)
  \Edge(3467)(134678)
  \Edge(2356)(234567)
  \Edge(2457)(234567)
  \Edge(3467)(234567)
  \Edge(12467)(12345678)
  \Edge(12357)(12345678)
  \Edge(13456)(12345678)
  \Edge(12348)(12345678)
  \Edge(23678)(12345678)
  \Edge(34578)(12345678)
  \Edge(24568)(12345678)
  \Edge(15678)(12345678)
  \Edge(123568)(12345678)
  \Edge(124578)(12345678)
  \Edge(134678)(12345678)
  \Edge(234567)(12345678)
}
\end{tikzpicture}}
\end{figure}
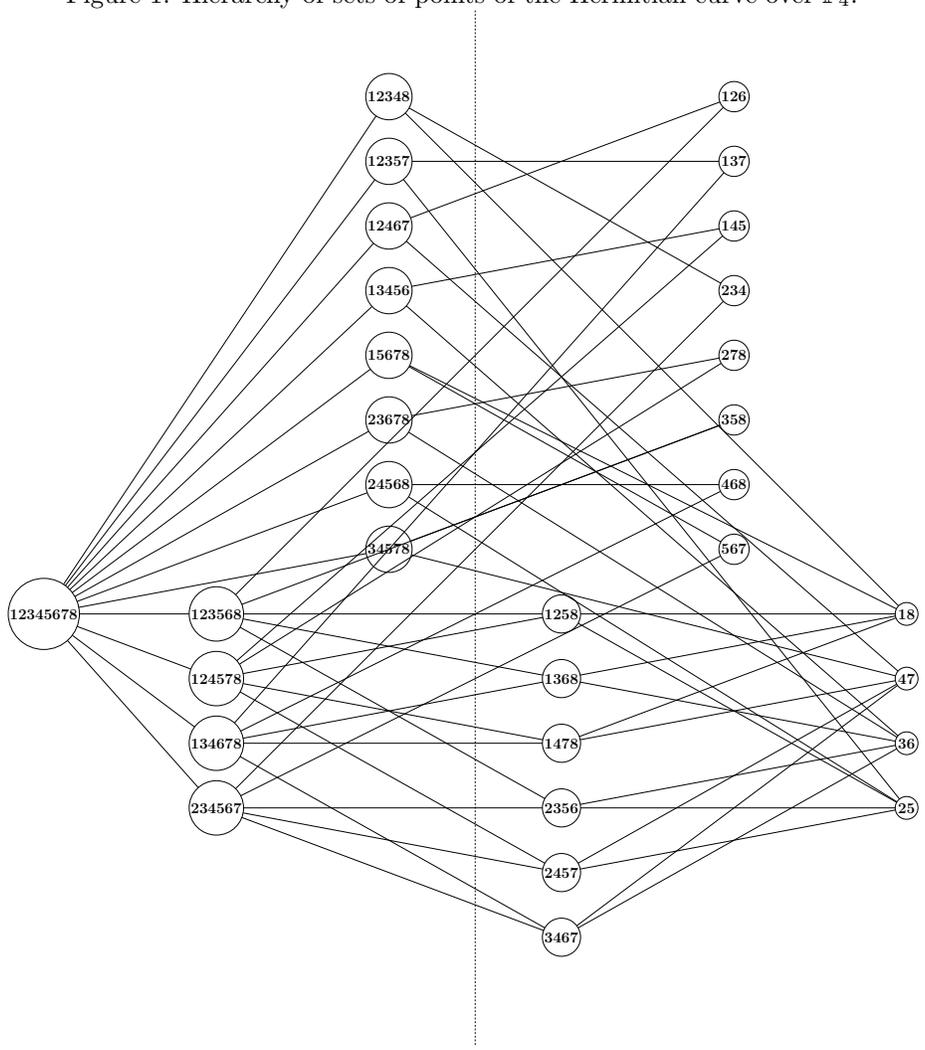

\section{Acknowledgment}
The author was supported by the Spanish government under grant TIN2016-80250-R and by the Catalan government under grant 2014 SGR 537.

\bibliographystyle{plain}

\begin{thebibliography}{1}

\bibitem{BLV}
Maria Bras-Amor\'os, Kwankyu Lee, and Albert Vico-Oton.
\newblock New lower bounds on the generalized {H}amming weights of {AG} codes.
\newblock {\em IEEE Trans. Inform. Theory}, 60(10):5930--5937, 2014.

\bibitem{GMRT}
Olav Geil, Carlos Munuera, Diego Ruano, and Fernando Torres.
\newblock On the order bounds for one-point {AG} codes.
\newblock {\em Adv. Math. Commun.}, 5(3):489--504, 2011.

\end{thebibliography}

\end{document}